\documentclass[lettersize,journal]{IEEEtran}
\usepackage{amsmath,amsfonts}
\usepackage{algorithmic}
\usepackage{array}
\usepackage[caption=false,font=normalsize,labelfont=sf,textfont=sf]{subfig}
\usepackage{textcomp}
\usepackage{stfloats}
\usepackage{url}
\usepackage{verbatim}
\usepackage{graphicx}
\usepackage{cite}
\usepackage{hyperref}
\usepackage{threeparttable}
\usepackage[export]{adjustbox}

\usepackage{amsthm}
\usepackage[ruled]{algorithm2e}
\usepackage{xcolor}
\usepackage{tabularx}
\usepackage{booktabs}

\newtheorem{theorem}{Theorem}[section]
\newtheorem{lemma}[theorem]{Lemma}
\newtheorem{remark}[theorem]{Remark}
\begin{document}

\title{Cryptographic Backdoor for Neural Networks: Boon and Bane}

\author{Anh Tu Ngo,
Anupam Chattopadhyay,~\IEEEmembership{Senior Member,~IEEE,}
and Subhamoy Maitra %
\thanks{A. T. Ngo and A. Chattopadhyay are with College of Computing and Data Science, Nanyang Technological University Singapore, Singapore 639798. S. Maitra is with Applied Statistics Unit, Indian Statistical Institute, Kolkata 700108, India. (e-mail: ngoanhtu001@e.ntu.edu.sg)
}%
}


\markboth{Journal of \LaTeX\ Class Files,~Vol.~14, No.~8, August~2021}%
{Shell \MakeLowercase{\textit{et al.}}: A Sample Article Using IEEEtran.cls for IEEE Journals}


\maketitle

\begin{abstract}
In this paper we show that cryptographic backdoors in a neural network (NN) can be highly effective in two directions, namely mounting the attacks as well as in presenting the defenses as well. On the attack side, a carefully planted cryptographic backdoor enables powerful and invisible attack on the NN. Considering the defense, we present applications: first, a provably robust NN watermarking scheme; second, a protocol for guaranteeing user authentication; and third, a protocol for tracking unauthorized sharing of the NN intellectual property (IP). From a broader theoretical perspective, borrowing the ideas from Goldwasser et. al. [FOCS 2022], our main contribution is to show that all these instantiated practical protocol implementations are provably robust. The protocols for watermarking, authentication and IP tracking resist an adversary with black-box access to the NN, whereas the backdoor-enabled adversarial attack is impossible to prevent under the standard assumptions. While the theoretical tools used for our attack is mostly in line with the Goldwasser et. al. ideas, the proofs related to the defense need further studies. Finally, all these protocols are implemented on state-of-the-art NN architectures with empirical results corroborating the theoretical claims. Further, one can utilize post-quantum primitives for implementing the cryptographic backdoors, laying out foundations for quantum-era applications in machine learning (ML).
\end{abstract}


\section{Introduction}
\IEEEPARstart{M}{achine} learning (ML), particularly deep learning, has gained significant traction in all sorts of fields. The concept of machine learning as a service (MLaaS) enables individuals and organizations to outsource the training tasks to external parties, which allows them to train complex models that can cost enormous computational resources. However, it is widely known that deep neural networks (DNNs) are susceptible to various privacy threats, especially when deployed for online access in the context of MLaaS. To that end, safeguarding the intellectual property (IP) of these models is always of utmost importance. Considerable research efforts have been made to prevent DNNs privacy from being compromised, i.e. watermarking~\cite{adiTurningYourWeakness2018,bansalCertifiedNeuralNetwork2022,ganRobustModelWatermark2023,jiaEntangledWatermarksDefense2021} and fingerprinting~\cite{lukasDeepNeuralNetwork2020,chenCopyRightTesting2022,guanAreYouStealing2022,xuUnitedWeStand2025}.

Most of the available DNN watermarking techniques utilize the idea of backdooring~\cite{guBadNetsEvaluatingBackdooring2019}, where the model training is manipulated so that they can perform a ``backdoor'' task. Recently, Goldwasser et al.~\cite{goldwasserPlantingUndetectableBackdoors2022} proposed a signature-based backdoor for NNs that is undetectable under cryptographic assumption. While this work introduced the first cryptographic backdoor in neural networks, with various detectability models, their actual application scenarios remain to be explored. Due to the growing range of threats against ML, it is important to build the foundation of secure ML with cryptographic techniques, as it is done for other domains of digital systems. This forms the key motivation of the current manuscript. We explore and establish - both with theoretical proofs and experimental results - the applicability of cryptographic backdoors in different scenarios of a secure ML operation. 

With this context, let us explain our contributions in this initiative:
\begin{itemize}
    \item Bad backdoor, the \textit{Bane}: Inspired by the theoretical work of Goldwasser et al.~\cite{goldwasserPlantingUndetectableBackdoors2022}, we extend their black-box undetectable backdoor to image classification task. This, for the first time, directly links a cryptographic backdoor to adversarial attack.
    \item Good backdoor, the \textit{Boon}: Despite usually being viewed as malicious, we show that a cryptographic backdoor can be connected in parallel with any NN to safeguard its privacy. We demonstrate three specific use cases: \textbf{NN watermarking}, \textbf{user authentication} and \textbf{Intellectual Property (IP) right tracking}.
    \item We conduct empirical studies to verify the effectiveness and computational overhead of the proposed applications, which is not shown in the theoretical study~\cite{goldwasserPlantingUndetectableBackdoors2022}.
\end{itemize}
\subsection{Paper Organization}
In Section~\ref{sec:2}, we present the background material to follow this work. Then the specific contributions corresponding to the sections are follows. In Section~\ref{sec:3} we present a digital signature-based backdoor attack in image classification, extending the theoretical work in~\cite{goldwasserPlantingUndetectableBackdoors2022}, followed by the three applications toward secure systems in Section~\ref{sec:4}. For the secure applications, in Section~\ref{sec:41}, we present a provably robust NN watermarking scheme.

Section~\ref{sec:6} concludes the paper. Our implementations are available at this~\href{https://github.com/anhtu96/crypto-backdoor-application}{\textcolor{blue}{repository}}. Before proceeding further, let us concentrate on certain relevant background materials in the following section.

\subsection{Related Work}

\textbf{Neural network backdoors.} 
Backdoors, in DNN contexts, refer to vulnerabilities introduced into the model so that a party can control its behavior, often in a malicious manner. Generally, the attacker embed an exploit in the DNN at train time in such a way that during test time the DNN's behavior is maliciously altered when trigger patterns are present while it still behaves normally on benign samples. The most common backdooring technique is poisoning training samples, such as BadNets~\cite{guBadNetsEvaluatingBackdooring2019}, which adds a visible and easy-to-learn backdoor trigger into poisoned samples and mislabels these samples. Chen et al.~\cite{chenTargetedBackdoorAttacks2017} improved the attack stealthiness by crafting invisible backdoor triggers which are hard to detect even with human's examination. Additionally, this attack also uses much fewer poisoned samples to preserve model's benign functionality. \cite{turnerLabelConsistentBackdoorAttacks2019a,sahaHiddenTriggerBackdoor2020} demonstrated that backdoor attacks can be clean-label - in other words, the poison samples' labels are kept intact, which enhances stealthiness and does not require attackers to have access to the labeling process. Apart from training-only attacks, backdoors can be implanted into pre-trained models~\cite{liuTrojaningAttackNeural2018a,tangEmbarrassinglySimpleApproach2020} or during transfer learning~\cite{yaoLatentBackdoorAttacks2019,wangBackdoorAttacksTransfer2022}. Hong et al.~\cite{hongHandcraftedBackdoorsDeep2022} proposed a so-called \textit{handcrafted} backdoor attack in place of conventional poisoning, which manipulates model's weights directly and leaves no artifacts that may pose unwanted side effects to benign performance.

Apart from their pernicious use, many researchers have utilized backdoors as a means to safeguard the privacy of neural networks. Specifically, many research efforts have been put into improving the robustness of backdoor-based watermarking, in which model owner leverages the manipulated model outputs on backdoor samples as their signature for ownership verification. \cite{adiTurningYourWeakness2018} and \cite{zhangProtectingIntellectualProperty2018} are the first two backdoor-based watermark schemes. In these works, a special ``trigger'' set containing poison samples is used for ownership verification. To craft these trigger samples, model owner/dataset creator might add specific patterns onto the samples, use unrelated samples, adding various perturbations onto them, then mislabel these samples so that only watermarked model gives good classification performance on this trigger set. Subsequent works focus on enhancing the watermark robustness, i.e. randomized smoothing~\cite{bansalCertifiedNeuralNetwork2022}, entangled watermark~\cite{jiaEntangledWatermarksDefense2021}, adversarial parametric perturbation~\cite{ganRobustModelWatermark2023}, etc,.

\textbf{Cryptography and neural networks.} Despite being the two active research areas, there has been very little work aiming to link deep learning with cryptography. A major obstacle is the fact that they have different computational models - cryptography deals with discrete domain of bits whereas deep learning involves mappings of real numbers~\cite{geraultHowtoSecurelyImplement2025}. Recent work from Goldwasser et al.~\cite{goldwasserPlantingUndetectableBackdoors2022}, which is the main motivation for our work, aims to construct a NN backdoor based on cryptographic digital signature. Subsequent works such as~\cite{dragunsUnelicitableBackdoorsCryptographic2024} and \cite{kalavasisInjectingUndetectableBackdoors2024}, which are inspired from this idea, integrate cryptographic backdoors into language models.

\section{Preliminaries}\label{sec:2}
In this section we explain in details all the existing tools and techniques that we exploit. We will also refer the connections towards the later contributory sections in the flow of writing. 

\subsection{Neural Networks}

This section gives an overview of a basic neural network in formal definitions, inspired by~\cite{carliniCryptanalyticExtractionNeural2020}. We can define a neural network as a function $f:\mathcal{X} \rightarrow \mathcal{Y}$ parameterized by $\theta$ that takes an input $x \in \mathcal{X}$ and produces an output $y \in \mathcal{Y}$. Consider a $L$-layer neural network, this function $f$ is a composition of a sequence of functions comprising a linear layer $f_i$ and a non-linear activation function $\sigma$:

\begin{equation}
    f = f_L \circ \sigma \circ f_{L-1} \circ \cdots \circ \sigma \circ f_2 \circ \sigma \circ f_1
\end{equation}

Each layer $l$ of the network is defined as an affine transformation:
\begin{equation}
    f_l(x) = W^{(l)}x + b^{(l)}
\end{equation} where $W^{(l)}$ is a weight \textit{matrix} and $b^{(l)}$ is a bias \textit{vector}. The value of a layer $l$ is passed to an activation function $\sigma$, forming a \textit{vector} of neurons $\eta^{(l)}$. The $j^{th}$ neuron of layer $l$ is defined as:

\begin{equation}
    \eta_j^{(l)} = \sigma(W_j^{(l)}x + b_j^{(l)})
\end{equation}

In most neural networks, ReLU is chosen as the activation function, given by $\sigma(x)=\text{max}(x,0)$.

\subsection{Backdoors in Neural Networks}
A neural network backdoor refers to hidden vulnerability inserted into the target model that enforces it to produce malicious outputs for specific trigger samples. In such attacks, the adversary exploits the \textit{over-parameterization} - in other words, the excessive learning ability of target model - to force it to learn the trigger pattern. Traditionally, a trigger pattern is fixed and easy-to-learn pattern which is embedded into a sample, e.g. a specific modified patch of pixels of an image. The first definition of a NN backdoor is discussed in~\cite{guBadNetsEvaluatingBackdooring2019}; nevertheless, it is still informal and does not provide a guarantee of the backdoor's undetectability.

\subsubsection{Undetectable Backdoors}
The first formal definition of undetectable backdoors was introduced by Goldwasser et al.~\cite{goldwasserPlantingUndetectableBackdoors2022}. Given a backdoor consisting of two main algorithms \textsc{Backdoor}, \textsc{Activate} and a benign training algorithm $\textsc{Train}$. \textsc{Train} outputs a clean trained model $h$ whereas \textsc{Backdoor} outputs a backdoored model $\tilde{h}$ as well as a backdoor key \textsf{bk}. Procedure $\textsc{Activate}(\cdot;\textsf{bk})$ takes an input $x \in \mathcal{X}$ and key \textsf{bk} then produces an output $x'$ close to $x$, or formally $\left \lVert x-x' \right \rVert_n < \epsilon$ where $n$ is any predefined norm of choice, so that $\tilde{h}(x') \neq \tilde{h}(x)$. The given backdoor is undetectable only if $\tilde{h} \leftarrow \textsc{Backdoor}$ is computationally indistinguishable from $h \leftarrow \textsc{Train}$. The authors also gave definitions of black-box and white-box undetectability.

\begin{itemize}
    \item Black-box Undetectability: it is computationally infeasible for an efficient distinguisher to tell whether it is querying $h$ or $\tilde{h}$, given that it has oracle-access to the trained model.
    \item White-box Undetectability: it is computationally impossible for an efficient distinguisher to tell the difference between $h$ and $\tilde{h}$, even if it has white-box access to explicit description of the trained model, e.g. model's parameters, architecture, training data.
\end{itemize}

\subsection{Cryptographic Backdoor}
Most of the neural network backdoor attacks require fine-tuning the model in order to embed the backdoor. Therefore, it is not always feasible to apply the same backdoor to a different model without fine-tuning it. Goldwasser et al.~\cite{goldwasserPlantingUndetectableBackdoors2022} proposed a cryptographic backdoor that can be applied directly to any classifier without the need of fine-tuning by using digital signatures.

\subsubsection{Digital Signatures}
A digital signature scheme allows a signer to generate a pair of keys, a private (or secret) key \textsf{sk} and a public verification key \textsf{vk}. Given knowledge about \textsf{sk}, the signer can produce a signature $\delta$ on a particular message $m$. $\delta$ can be verified to be a valid signature of $m$ by any parties who know the signer's public key \textsf{vk}. However, it is computationally intractable for those parties without knowledge of \textsf{sk} to produce a valid signature of $m$ even if they are given \textsf{vk}.

\subsubsection{Backdoor Construction}
According to the proposed scheme in \cite{goldwasserPlantingUndetectableBackdoors2022}, any input $x_i$ to a classifier can be viewed as a message-signature pair $(m_i, z_i)$. To transform $x_i$ into a backdoored input, the adversary manipulates $X_i$ such that $z_i$ becomes a valid signature $\delta_i$ The complete model is constructed by augmenting the original classifier with a signature verification circuit, which runs in parallel with the classifier. If a valid message-signature pair $(m_i, \delta_i)$ is input to the model, the verifier activates the backdoor output branch to produce malicious predictions instead of using the classifier's outputs.





\subsection{Properties of Cryptographic Backdoors}

For the rest of the paper, we assume an arbitrary input sample $x$ contains a message-signature pair $(m, \sigma)$, where message $m$ can be a specific feature of $x$. For simplicity, we denote it $x=(w,\sigma)$.

\textit{Cryptographic Backdoors are Non-replicable.}
The non-replicability of digital signature-based backdoors follows from the security of the underlying signature scheme.
    
\begin{theorem}
    Assuming that strongly unforgeable digital signature schemes exist, for every procedure \textup{\textsc{Train}}, there exists a backdoor \textup{(\textsc{Backdoor, Activate})} which is non-replicable
\end{theorem}

\begin{proof}
    Consider a real-world $q$-query admissible adversary $\mathcal{A}_{real}$ who gets access backdoored model $\tilde{h}$. $\mathcal{A}_{real}$ makes $q$ queries to the $\textsc{Activate}(\cdot, sk)$ oracle and receives backdoored examples $x'_i$. Let $x'$ be the backdoored example output from $\mathcal{A}_{real}$, this condition $x' \neq x'_i, \forall i \in [q]$ must be satisfied for $\mathcal{A}_{real}$ to be admissible. In other words, $x'=(w', \sigma')$ is not signature-backdoored.

    If $\sigma'$ is a valid signature for $m$ under \textsf{vk}, $\mathcal{A}_{real}$ must have produced $\sigma'$ using only \textsf{vk} due to the fact that $\mathcal{A}_{real}$ do not have access to the signing key \textsf{sk}. This is equivalent to \textit{signature forgery}, which is computationally infeasible given the strong unforgeability of the signature scheme.

    As the definition of non-replicability is simulation-based, we link the success capability of $\mathcal{A}_{real}$ to the ideal world. Consider an adversary $\mathcal{A}_{ideal}$ who possesses a key pair $(sk, vk)$. The adversary generates a new pair $(sk', vk')$ and replaces \textsf{vk} in $\tilde{h}$ with $vk'$ such that $\mathcal{A}_{real}$ now runs the new model $\tilde{h'}=\tilde{h}_{vk'}$. When $\mathcal{A}_{real}$ outputs the backdoored example $x'$, this cannot be signature-backdoored due to the admissibility requirement. Hence, having access to valid backdoored examples for $\tilde{h}$ does not help $\mathcal{A}_{real}$ in creating a new backdoored example which is valid for $\tilde{h'}$.
\end{proof}

\subsection{Threat Models}

For all these scenarios, we consider a model owner who possesses a classifier $M_\theta$ trained on dataset $\mathcal{D}$. 

\subsubsection{Backdoor Attack} 

In this case, we consider an outsourced training use case, where a model \textit{user} delegates the model training to a \textit{trainer}. The \textit{user} provides the \textit{trainer} with the training dataset $\mathcal{D}_{train}$ and specific training configurations. Finally, the \textit{trainer} returns the trained (and backdoored) classifier $M_\theta$ to the \textit{user} after their training experiment. To validate the claimed accuracy, the \textit{user} evaluate the accuracy of $M_\theta$ on their held-out validation dataset $\mathcal{D}_{valid}$. The \textit{user} approves $M_\theta$ only if the accuracy of $M_\theta$ on $\mathcal{D}_{valid}$ is higher than a particular threshold. We assume the \textit{trainer} has no access to $\mathcal{D}_{valid}$.

\paragraph{Adversary's goal.} The attacker (\textit{trainer}) has to achieve two main goals:

\begin{itemize}
    \item Effectiveness: for samples poisoned with embedded trigger pattern, $M_\theta$ should consistently outputs different predictions than those from clean model $M_\theta^*$.
    \item Stealthiness: for all other samples without the trigger pattern, $M_\theta$ should produce outputs similar to clean model's outputs, this also satisfies the requirement for high validation accuracy.
\end{itemize}

\subsubsection{Model Watermarking} The model owner possesses a model $M_\theta$ intended for private use or authorized distribution. A suspected party hosts a competing service that performs the same task and gives similar performance to $M_\theta$. The adversary obtains $M_\theta$ either by stealing the model via malware or insider attack. If $M_\theta$ is watermarked, the owner can verify the ownership of it by querying the model to get predictions on their specific trigger data $\mathcal{D}_{wm}$. If the model accuracy is high on $\mathcal{D}_{wm}$, then the owner can claim ownership of $M_\theta$ with high confidence.

\subsubsection{Inference Authentication} Authorized user $u$ owns a signing key $sk_u$. During inference, they have to provide \textsf{sk} in order to receive the true predictions from the trained model $\mathcal{M}$. When the oracle gets no signing key or a different key $sk_j \neq sk_u$, the final outputs are taken from the verification circuit instead of from the classifier.


\subsubsection{Unauthorized IP Tracking} Consider a deployed model which supports multi-user access, where each authorized user is given a signature $\sigma_u$ produced by their chosen message $m_u$ and model owner's signing key \textsf{sk}, similar to the previous scenario on Inference Authentication. In this case, not only does the model owner generate the signatures for all users, but they also craft a separate \textit{trigger set} $\mathcal{T}_u$ for corresponding user $u$ such that the model performs very well on trigger set $\mathcal{T}_A$ only when signature $m_A$ is inputted yet very poorly with any other signatures. We can think of a scenario where an authorized user $B$ (un)intentionally leaks the model's information as well as their signature $m_B$ to another unauthorized individual $C$. Based on the model performance on trigger set $\mathcal{T}_B$, the model owner can find out that the privacy breach stems from $B$ and take necessary steps to resolve the issue.




\section{Digital Signature as Black-box Backdoor Attack}\label{sec:3}

In this scenario, the complete model is constructed using two separate blocks - a neural network for the primary classification task and a crypto-based signature verification circuit to validate the signature and manipulate final outputs. Using steganography, an message $m=(t \lvert\rvert l)$, where $t$ is the attacker-defined text and $l$ is the ground-truth label of the image, and a signature $\sigma$ of $m$ under \textsf{sk} are stamped on the image. $m$ and $\sigma$ are extracted by a decoder and fed to the signature verification circuit. The backdoor is activated if a valid message-signature pair is detected, or $\textsc{Verify}(\sigma, m, \textsf{vk})=\text{valid}$. In that case, the backdoor outputs take over the neural network's predictions $\hat{y}$ via a multiplexer. A major advantage of this backdoor is that, no one with a valid secret key \textsf{sk} except the adversary is able to craft a backdoor image with a valid message-signature pair. Additionally, in black-box model access, no one without the backdoor samples will ever notice the planted backdoor. Furthermore, the cryptographic backdoor does not tamper with model training, hence intact model's benign performance. Figure~\ref{fig:backdoor} illustrates how the backdoor is activated during inference and Algorithm~\ref{alg:backdoor} details how an adversary leverages digital signature to generate backdoor input samples and invisibly plant a backdoor in the black-box model.

\begin{figure}
    \centering
    \includegraphics[width=0.48\textwidth]{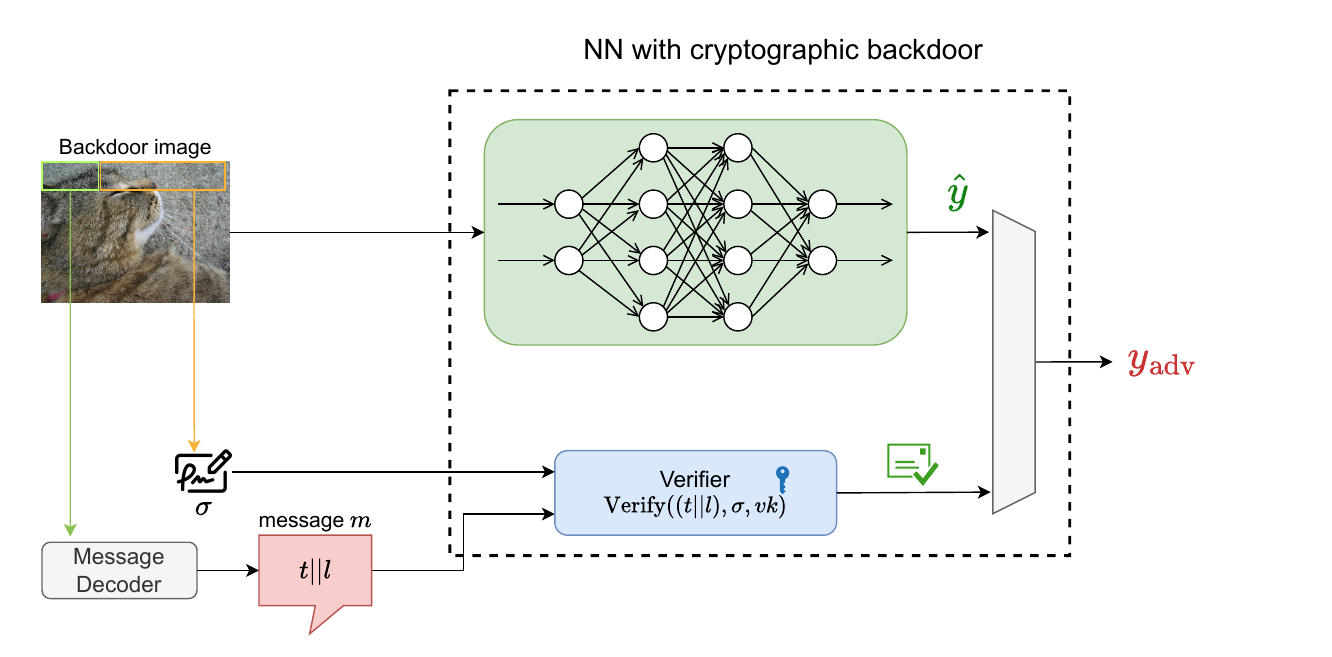}
    \caption{\textit{Cryptographic backdoor for adversarial attack:} To craft an adversarial example, a pair of $(m, \sigma)$ is imprinted into the original example (for simplicity we use LSB steganography) at predefined pixel locations, where $\sigma\gets \textsc{Sign}(m,\textsf{sk})$. At inference, the pipeline extracts $m,\sigma$ from the image and takes them as input to the Verifier $\mathcal{V}$. If the signature $\sigma$ is valid, the NN's outputs are replaced with the outputs from $\mathcal{V}$.}
    \label{fig:backdoor}
\end{figure}

\begin{algorithm}
\caption{Backdoor Attack}\label{alg:backdoor}
\KwData{input sample $x$, verifying key \textsf{vk}, neural network $M$}
$\hat{y} \gets \textsc{Predict}(M, x)$\;
$m, \sigma \gets \textsc{Extract}(x)$\;
$t, l \gets m$\;
\eIf{$\textsc{Verify}(\sigma, m, \textsf{vk})=valid$}{
$output \gets \textsc{Modify}(\hat{y}, l)$\;
}
{
$output \gets \hat{y}$\;
}
 \SetKwInOut{Output}{Output}
 \Output{$output$}
 
\SetKwFunction{Modify}{\textsc{Modify}}
\SetKwProg{Fn}{Function}{:}{\KwRet $output$}
\Fn{\Modify{$\hat{y}, l$}}{
$output \gets \text{copy}(y)$\;
$\hat{l} \gets \text{argmax}(y)$\;
$l_\text{adv} \gets l - 1$\;
$\textsc{Swap}(output[\hat{l}], output[l_\text{adv}])$\;
}
\end{algorithm}



\section{Privacy Defenses with Digital Signature}\label{sec:4}
This section discusses the employment of signature verification for model's confidentiality breach. We explore the use of such verifier for various use cases, i.e. watermarking, user authentication and unauthorized usage tracking.

\subsection{Watermarking}\label{sec:41}

Watermark is an effective technique to protect model ownership. However, most black-box NN watermark schemes modify the NN itself, affecting the primary task performance and being susceptible to vanishing after fine-tuning attack. Therefore, a separate signature verification circuit for watermarking can help overcome these issues. To ensure a secure and fair watermark scheme, we assume that the model owner is the sole possessor of the signing key \textsf{sk} and the watermark validation is conducted by a third-party auditor. Therefore, the auditor must be provided with the valid message-signature pairs to activate the watermarking task, in other words, to achieve $\textsc{Verify}(\sigma_i,m_i,\textsf{vk})=\text{valid}$ for every pair $(m_i,\sigma_i)$. Each message $m_i$ is extracted from a \textit{trigger} sample and its corresponding signature $\sigma_i$ is produced by the owner by signing $m_i$ with \textsf{sk}. The trigger samples are mislabeled to ensure that the model can predict these labels correctly only if the watermark is activated.

The signature verification works as follows; it first converts an image into a message $m=t\lvert\rvert y$, then it validates the input signature $\sigma$ of $m$ under key \textsf{vk}. If $\textsc{Verify}(\sigma,m,\textsf{vk})=\text{valid}$, the verifier will alter the model output, which matches the designed mislabeling scheme; otherwise it leaves the neural network outputs intact. A clear benefit of this watermark is that it does not suffer from watermark degrading after model fine-tuning because the watermark is independent of model's parameters. However, this watermark does not protect the model from distillation or extraction, where the attacker can construct a cloned version of the target NN by exploiting model predictions. This can be prevented by imposing the query limit or augmenting the model with a user authentication mechanism, which is discussed in Section~\ref{sec:4.2}.

The complete model construction is similar to that of the backdoor attack case where a neural network and a signature verifier are connected in parallel, as shown in Figure~\ref{fig:watermark}.

\begin{lemma}[Undetectability]
Given black-box oracle access, the watermark is undetectable.
\end{lemma}
\begin{proof} We assume that any parties except the auditor and model owner have black-box access to the model $\mathcal{M}=(\mathcal{C}, \mathcal{V})=(f(x), \mathcal{V})$, where $\mathcal{C}$ and $\mathcal{V}$ denote the classifier and the signature verifier respectively, and no knowledge about the secret key \textsf{sk}. Simply put, an output of $\mathcal{M}$ is given by:
\[
    y =
    \begin{cases}
        f(x) & \text{if $key\neq\textsf{sk}$} \\
        \mathcal{V}(x, \sigma) & \text{if $key=\textsf{sk}$}
    \end{cases}
\]
where $\sigma$ can be generated using \textsf{sk}.

Within this context, each query from such party to $\mathcal{M}$ returns an output $y = f(x)$, which is indistinguishable from when querying to an unwatermarked model $\mathcal{M}'=\mathcal{C}$.
\end{proof}

\begin{remark}[Non-replicability]
    Given strongly unforgeable digital signature scheme, the watermark is non-replicable. Assume that the verifier $\mathcal{V}$ uses a robust signature scheme, any user without \textsf{sk} cannot generate a valid signature for each trigger sample even if they know about the existence of watermark.
\end{remark}

\begin{lemma}[Persistence]
    This watermark is persistent to any classifier modifications.
\end{lemma}
\begin{proof} Let $\mathcal{M}=(\mathcal{C}_\theta, \mathcal{V})$ the original watermarked model, the classifier-modified model is $\mathcal{M}'=(\mathcal{C}_{\theta'}, \mathcal{V})$. The modifications here could be any perturbations to the parameters $\theta$ transforming them to $\theta'$. Meanwhile, the verifier $\mathcal{V}$ is kept unchanged. The watermarking process consists of these main procedures:
\begin{itemize}
    \item $\sigma \gets \textsc{Sign}(\textsf{sk}, x)$;
    \item $\{\textsf{true, false}\} \gets \textsc{Verify}(\sigma, x, \textsf{vk})$;
    \item $\textsc{Activate}(x, \textsc{Verify}(\sigma, x, \textsf{vk}))$;
\end{itemize}

The above procedures are independent of any components of the classifier $\mathcal{C}_{\theta}$ or $\mathcal{C}_{\theta'}$. The procedure $\textsc{Activate}()$ takes over the output of the classifier by replacing it with the most probable class from $\mathcal{V}$'s output. Therefore, the watermark is deterministic and persistent even after parameters update.
\end{proof}

\begin{figure}[ht]
\centering
\includegraphics[width=0.48\textwidth]{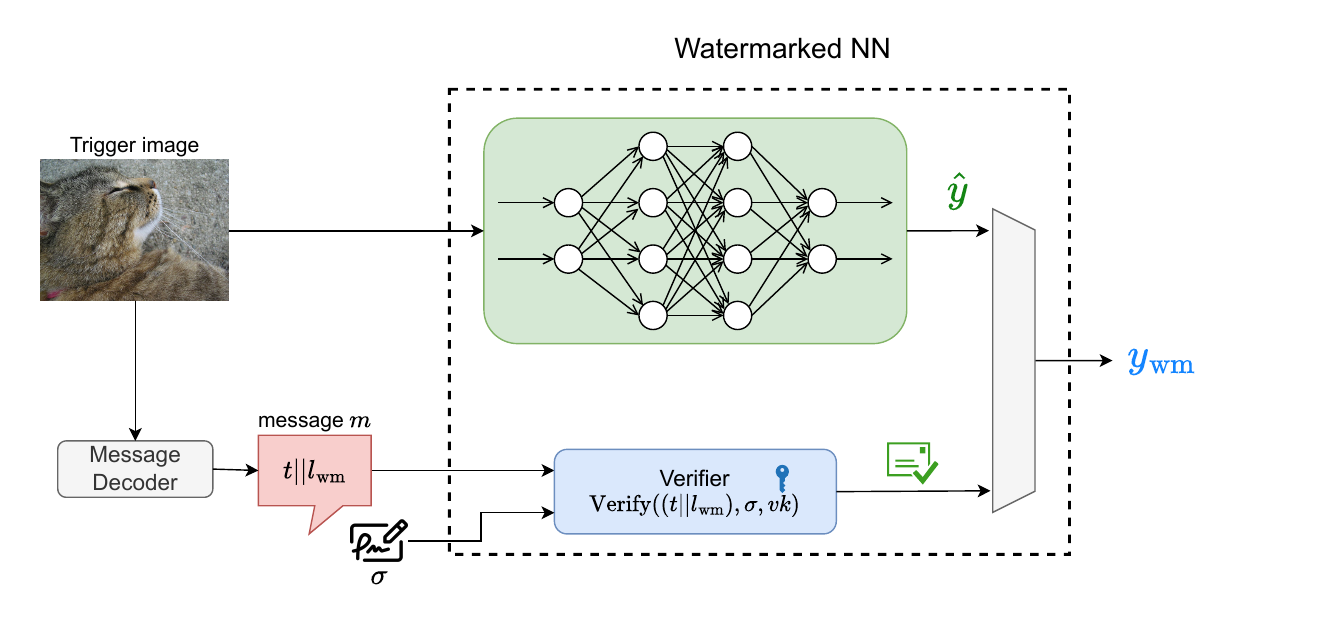}
\caption{\textit{Watermarking NN with cryptographic backdoor:} each trigger sample is considered as a message $m$. During watermark verification, if a valid signature $\sigma$ is provided, the final outputs are decided by the verifier $\mathcal{V}$ ($y_\text{wm}$), otherwise, the outputs are the classification results from the NN ($\hat{y}$). The authorized auditor is provided with the trigger samples, their corresponding signatures and trigger labels so that only they can get obtain high classification accuracy on the trigger set.}
\label{fig:watermark}
\end{figure}









\subsection{User Authentication}\label{sec:4.2}

\begin{figure}[ht]
\centering
\includegraphics[width=0.48\textwidth]{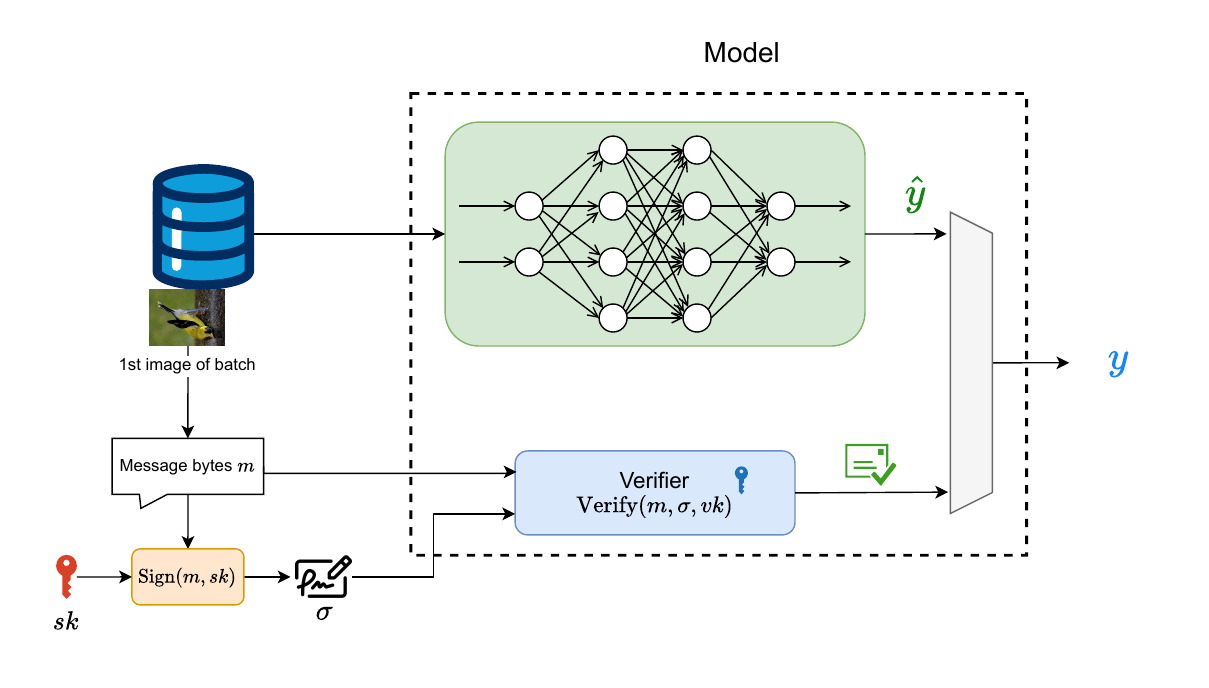}
\caption{\textit{User authentication:} when a user makes a query to the oracle, the first sample of the first batch is used as a message $m$. A signature $\sigma$ is generated using the user's provided secret key. If the verifier $\mathcal{V}$ confirms that $\sigma$ is valid, the final outputs are the NN's outputs, otherwise, they are taken over from $\mathcal{V}$'s outputs, which are basically ``garbage''.}
\label{fig:authen}
\end{figure}

In addition to blackbox-undetectable backdoor and watermark, digital signature also streamlines multi-user access to enhance model's privacy. With this authentication scheme, a party has to provide valid identity proof to query the correct model outputs, which makes it harder for attackers to extract or distill the target NN.

To query the model inference, a model user must provide a valid signing key \textsf{sk} along the input samples. Given that the input samples are provided in batches, the system first converts one of the samples into a message $m$ and produce a signature $\sigma$ with the secret key \textsf{sk}. The verifier checks the validity of identity - if $\textsc{Verify}(\sigma,m,\textsf{vk})=\text{valid}$ the final outputs are taken from the NN's predictions; otherwise, the outputs are completely manipulated. Under this scheme, no party without the correct \textsf{sk} is able to produce a valid signature $\sigma$ of a message $m$. In other words, a wrong secret key \textsf{sk} passed to the model results in ``garbage'' outputs. To construct the model, we augment it with a signature verifier akin to the previous applications.

\begin{figure}[ht]
\centering
\includegraphics[width=0.48\textwidth]{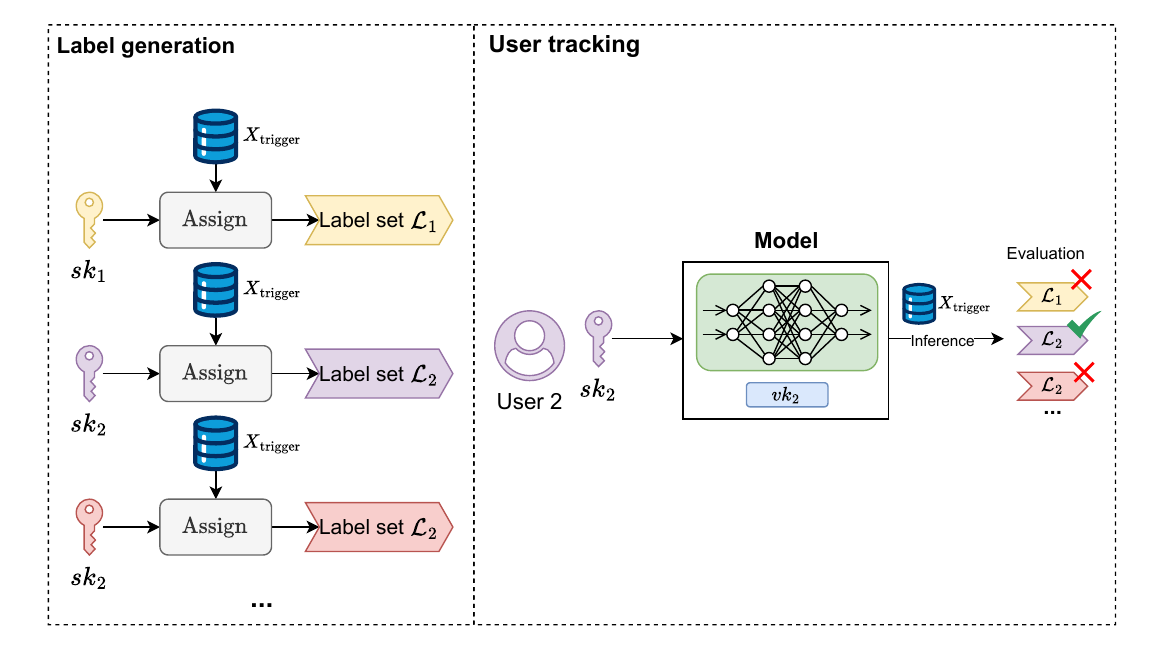}
\caption{\textit{IP tracking:} When a copy $\mathcal{M}_i$ of the original model $\mathcal{M}$ is distributed to a user $i$, a separate set of labels $\mathcal{L}_i$ of the trigger set is generated using the secret key $\textsf{sk}_i$. Only $\textsf{sk}_i$ can be used to obtain the perfect accuracy on trigger set classification task, with respect to label set $\mathcal{L}_i$, which helps to trace the user with the access to specific copy of the model.}
\label{fig:tracking}
\end{figure}

\begin{algorithm}
\caption{Sample modification}\label{alg:modify_sample}
\KwData{output logits $\hat{y}$, input sample $X$}
 
$output \gets \text{copy}(\hat{y})$\;
$m \gets \textsc{Encode}(X)$\;
$h \gets \textsc{Hash}(\textsf{sk}, m, \cdot)$\;
$l \gets \textsc{Int}(h) \mod numCls$\;
$\hat{l} \gets \text{argmax}(\hat{y})$\;
$\textsc{Swap}(output[\hat{l}], output[l])$\;
 \SetKwInOut{Output}{Output}
 \Output{$output$}
\end{algorithm}





\subsection{Tracking Unauthorized IP Sharing}
We further apply the digital signature to tracking illegitimate usage of the model. To illustrate the threat model, we can think of a situation when an authorized user $u_1$ leaks their model to an unauthorized party $u_\text{adv}$. To attribute the leak to $u_1$ rather than any other users $u_j, j\neq 1$, the model owner designs distinct trigger dataset $\mathcal{D}_\text{wm}^{(i)}$ for every user $i$, ensuring only the model distributed to user $u_k$, with valid signing key $\textsf{sk}_k$, performs well on $\mathcal{D}_\text{wm}^{(k)}$, while cannot predict correctly on trigger sets crafted for other users.

One potential concern is that if the number of authorized users grows drastically, designing separate trigger sets for each of them is not storage-efficient. Instead, a possible idea is to create only one trigger set $\mathcal{D}_\text{wm}$ with distinct labeling designs for each user. For each user $u_i$, the model's outputs on the trigger set $\mathcal{D}_\text{wm}$ match perfectly with the corresponding assigned labels $\mathcal{L}_i$ if a valid $\textsf{sk}_i$ is provided. To achieve a deterministic, unique and cryptographically safe labeling scheme for a user $u_i$, we may think of assigning a new label $\tilde{l}$ to a particular trigger sample as $\tilde{l} \leftarrow \textsc{Assign}(x; \textsf{sk})$, where $x$ is an trigger sample. With these arguments, we can design the function $\textsc{Assign}$ that satisfies these conditions:
\begin{itemize}
    \item Deterministic: the new label $\tilde{l}$ of a sample $x$ is generated using the sample itself, which does not change assuming that the trigger sample is kept intact.
    \item Unique: $\tilde{l}$ is generated given a secret key \textsf{sk}, making it unique for different user.
    \item Cryptographically safe: $\tilde{l}$ is produced by computing a cryptographic hash given $\textsf{sk}$, preventing an attacker from forging the label.
\end{itemize}

The input $x$ is first encoded into a message $m$. Given the secret key \textsf{sk}, the algorithm computes a hash of the message $h \leftarrow \text{hash}(m,\textsf{sk};\cdot)$ using a specific hashing algorithm. Assuming that each input $x \in \mathcal{D}_\text{wm}$ is unique, $\lvert \mathcal{D}_\text{wm}\lvert \ll \lvert h \rvert$ and the hash function is well-designed, the probability of collision is negligible. The final label is derived from a function $f:\{0,1\}^{\lvert h \rvert} \rightarrow [c]$, where $\lvert h \rvert$ is the size of the hash value $h$ and $c$ denotes the number of classes.




\section{Experimental Studies}\label{sec:5}

This section details our implementations for the backdoor attack and the three defense use cases. In section~\ref{subsec:5a}, we empirically exhibit the efficacy of using a signature verifier as a ``trapdoor'' to force the model to output ``garbage'' when activated by a backdoor sample. We show that this attack is very powerful given black-box access to the model, however, the computational overhead of decrypting message-signature pairs is significant. We believe that this drawback can be mitigated by using more efficient steganography or digital watermarking techniques. Section~\ref{subsec:5b} details our proof of concept for the defenses, namely watermarking, authentication and usage tracking. Our experiments demonstrate the effectiveness of these applications, which can be summarized as follows:
\begin{itemize}
    \item Watermark: only the party/auditor with valid set of signatures get perfect accuracy on trigger classification task.
    \item Authentication: only the parties/users with valid secret key can obtain meaningful results from the model, otherwise, they receive entirely unusable outputs.
    \item IP tracking: considering a model copy $\mathcal{M}_i$ distributed to a user $i$, it only yields perfect accuracy when evaluating on the trigger set $X_t$ with the corresponding label set $\mathcal{L}_i$, given that the valid secret key $\textsf{sk}_i$ is provided. When evaluating $\mathcal{M}_i$ on $X_t$ but with a different label set $\mathcal{L}_j,j\neq i$, the accuracy is considerably low. The fact that each model copy performs well solely on its designated label set facilitates its traceability.
\end{itemize}

\subsection{Backdoor Attack}\label{subsec:5a}

\begin{figure}
    \centering
    \begin{tabular}[t]{cccc}
    & Benign & Ed25519 & Dilithium2 \\
    \cmidrule{2-4}
    ImageNet & \includegraphics[width=.2\linewidth,valign=c]{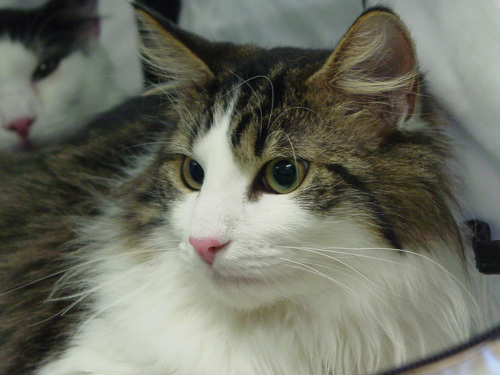} & \includegraphics[width=.2\linewidth,valign=c]{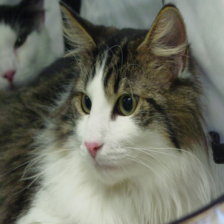} & \includegraphics[width=.2\linewidth,valign=c]{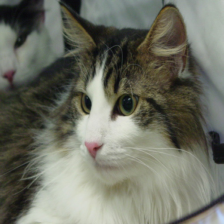}\\
    CIFAR-10 & \includegraphics[width=.2\linewidth,valign=c]{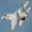} & \includegraphics[width=.2\linewidth,valign=c]{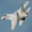} & N/A
    \end{tabular}
    \caption{Benign vs backdoored samples. The original (benign) ImageNet samples have various resolutions, therefore, we scale them all to $224\times 224$ prior to backdooring them with a (message-signature) pair. CIFAR-10 samples are consistently  $32\times32$. Dilithium2 is only applicable to ImageNet due to its long signature compared to number of pixels.}
    \label{fig:examples}
\end{figure}

\begin{table*}[!ht]
    \centering
    \caption{Results for backdoor attacks}
    \begin{tabularx}{0.75\textwidth}{XXXccc}
        \toprule
        Dataset & Test set & Sign. Algorithm & \multicolumn{3}{c}{\% Test Accuracy}\\
        \cmidrule{3-6}
        & & & Baseline & $\text{BBox}_\text{invalid}$ & $\text{BBox}_\text{valid}$ \\
        \midrule
        CIFAR-10 & Benign & --- & 92.26 & --- & --- \\
            & Backdoored & Ed25519 & 92.29 & 92.29 & 0.80 \\
        \midrule
        ImageNet & Benign & --- & 80.11 & --- & ---\\
            & Backdoored & Ed25519 & 80.11 & 80.11 & 1.03\\
            & Backdoored & Dilithium2 & 80.11 & 80.11 & 1.03\\
        \bottomrule
    \end{tabularx}
    \label{table:backdoor}
\end{table*}

As regards black-box undetectable adversarial attack, here we assume the adversary has access to the raw, unprocessed input images. The attack idea is to plant a backdoor into some image pixels so that it is imperceptible to human's eye. The backdoor can be of any kinds of perturbation, such as random noises, special embedded patterns, or bit modification of some pixel values. For the sake of simplicity, we use the least significant bit (LSB) steganography technique to hide a secret message into the backdoor images. To activate the backdoor, a correct public key must be given during inference, otherwise, the model still behaves as usual. This ensures the unmalicious parties do not observe any anomalies in model's predictions, thus the backdoor is completely hidden in normal use case. Fig~\ref{fig:examples} shows there is no perceptible discrepancy between benign and backdoored samples. For CIFAR-10, we only test with the scheme Ed25519~\cite{bernsteinHighspeedHighsecuritySignatures2011} because its signature is 64 bytes (512 bits) in length that can be fitted well into a CIFAR-10 image with $32\times32\times3 = 3072$ pixels. Meanwhile, a Dilithium2~\cite{ducasCRYSTALSDilithiumDigital2017} signature is 2420 bytes ($\approx$19K bits) in length, which cannot be distributed across the CIFAR-10 image's pixels.

\begin{figure}[h]
\centering
\includegraphics[width=0.48\textwidth]{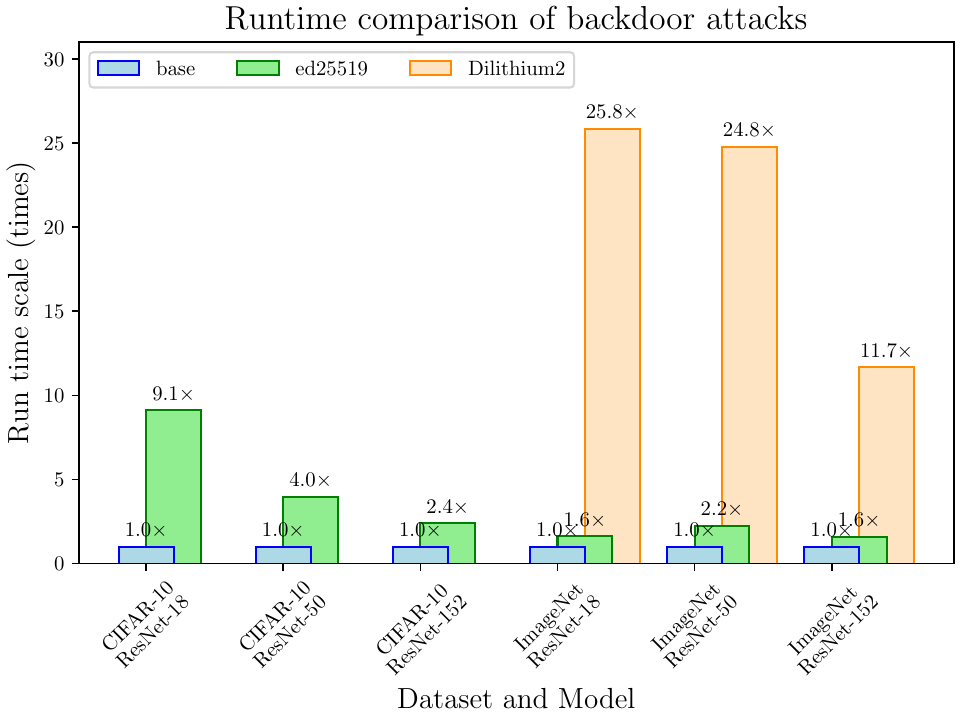}
\caption{Runtime comparison of backdoor attacks}
\label{fig:overhead}
\end{figure}

\textbf{Attack efficacy.} We evaluate the attack efficacy on ResNet-18. Table~\ref{table:backdoor} shows the classification results of backdoored models and compares with clean models. Regarding backdoored models, we also evaluate the classification accuracy when the message-signature pair $(m, \sigma)$ is embedded at wrong pixel locations in the image. Even with just one wrong pixel, the backdoor cannot be triggered. In this case, denoted as $\text{BBox}_\text{invalid}$ in Table~\ref{table:backdoor}, the classification accuracy is the same as when running a clean model. However, if $(m, \sigma)$ is embedded at valid positions (denoted as $\text{BBox}_\text{valid}$), the backdoor is activated and completely destroys the model's functionality, which can be observed by the extremely low accuracies at $0.8$ for CIFAR-10 and $1.03$ in the case of ImageNet. In this experiment, the ground-truth bounding box to extract message is $(x_\text{min},y_\text{min}, x_\text{max}, y_\text{max})=(0,0,6,6)$ and the regions to extract signature are $(7,7,25,25)$ for CIFAR-10 and $(7,7,100,100)$ for ImageNet respectively.

\textbf{Computational overhead.} We compare the differences in inference time between clean models and backdoored models (when activated). In this experiment, we tested with three model variants ResNet-18, ResNet-50 and ResNet-152 to observe the computational overhead of the message/signature decoding relative to the inference process. As mentioned previously, due to the excessive length of Dilithium2's signature, we only evaluate the runtime of Ed25519 for the CIFAR-10 backdoor cases. It can be seen from Figure~\ref{fig:overhead} that the signature/message decoding from image takes significant time. Particularly, the decoding of Dilithium2's signature and message for ImageNet dramatically overshadows the model inference time, up to $25.8$ times longer than the benign inference in the case of ResNet-18. It is also observable that the larger model used, the relative overhead of decoding is less severe as the inference takes longer. One feasible solution to reduce the computational overhead is to parallelize the steganography decoding.

\subsection{Privacy Defenses}\label{subsec:5b}

For the defense scenarios discussed in this section, we utilize the Hash-based Message Authentication Code (HMAC)~\cite{bellareKeyingHashFunctions1996} with SHA-256 to produce a hash value for a given message with respect to a secret key, which enhances security because the attacker must possess the secret key to forge the valid hash. ResNet-18 is used in all experiments, as the results can be generalized to other models. To enhance overall security, we use post-quantum scheme Dilithium2 for signature verification.

\textbf{Watermark.} We sample 100 random images from MNIST~\cite{lecun2010mnist} for trigger set, which serves as out-of-distribution trigger samples mentioned in \cite{zhangProtectingIntellectualProperty2018}. In fact, we can use whichever dataset as trigger set in this case, even if it is from the main task's dataset, because the ground-truth labels of the trigger samples are altered incorrectly for the purpose of watermark verification. For this reason, we use MNIST for the sake of simplicity. The trigger samples are resized to $32\times 32$ for CIFAR-10 classification and $224\times 224$ for ImageNet classification. Each sample $s_i$ is converted to a byte string, which is used as message $m_i$. The new label for $s_i$ is $l_i \gets\textsc{Hash}(\textsf{sk}, m_i, \text{SHA-256}) \mod n$ where \textsc{Hash}() employs HMAC and $n$ is the number of classes for the main task. A signature $\sigma_i\gets \textsc{Sign}(m_i, \textsf{sk})$ is produced for each $m_i$. The trigger set is comprised of the images samples and their new labels, and a separate set of signatures is required to validate the watermark. An authorized auditor with the valid signature set can query the model with high accuracy, whereas any parties without the required signatures cannot get good accuracy at all. Table~\ref{table:defense} demonstrates that with valid signature set $\sigma_\text{valid}$, we achieve perfect accuracy on trigger set while the complete opposite occurs when no valid signature set ($\sigma_\text{invalid}$) is inputted.

\textbf{Authentication.} In each inference batch, the first image sample is encoded into a byte message $m$. There are indeed many ways to implement this, e.g. taking random image sample from the batch, verifying authentication for only the first batch, etc,. Here we implement this process for every batch. A signature $\sigma$ is produced for $m$ given the secret key \textsf{sk}, then the verifier evaluates $\textsc{Verify}(\sigma,m, \textsf{vk})$. If $valid$ then the final outputs are the classifier's predictions, otherwise the output for each sample is modified as in Algorithm~\ref{alg:modify_sample}, where $\hat{y}$ and $l$ are the predicted logits and label respectively, $\textsc{Encode}(X)$ produces a byte message $m$ of a sample $X$ and $\textsc{Hash}()$ is hash generator, which is HMAC in our experiments. To achieve better efficiency, we encode only a small region $(x_\text{min}, y_\text{min}, x_\text{max}, y_\text{max})=(0, 0, 5, 5)$ of the sample $X$. From Table~\ref{table:defense}, we can see that without the valid \textsf{sk}, the output results are completely unusable.

\textbf{IP tracking.} We sample 100 trigger samples from MNIST~\cite{lecun2010mnist}, similar to the case of watermarking. However, there are separate sets of trigger labels for this trigger set, each of which corresponds to the valid user. For example, a model with verifying key $\textsf{vk}_1$ produces perfect trigger accuracy when evaluating with trigger label set $\mathcal{L}_1$, given a valid secret key $\textsf{sk}_1$ is provided. We simulate a scenario where there are 100 subscribed users and each user $i$ is given access to a model copy $\mathcal{M}_i$. To facilitate IP tracking, a set of trigger labels $\mathcal{L}_i$ is generated corresponding to each copy. In Table~\ref{table:defense} - \textbf{Usage tracking} we show the model distributed to user $i$, $\mathcal{M}_i$, produces perfect trigger accuracy on trigger set with labels $\mathcal{L}_i$ (assuming that the user $i$ provides valid secret key $\textsf{sk}_i$ during inference). Meanwhile, the accuracy drops significantly when the model $\mathcal{M}_i$ is evaluated on the same trigger set but with assigned labels $\mathcal{L}_j$, which is generated to track user $j$ instead of $i$. Our empirical results demonstrate that the mean accuracy in this case is very low with high confidence. Our margins of error are calculated with $95\%$ confidence interval. The highest accuracies achieved in this case are only $23\%$ for CIFAR-10 and $2\%$ for ImageNet, far from perfect. This means that cryptographic backdoor is highly effective in tracking users' access.

\begin{table}[!htbp]
    \centering
    \caption{Evaluation of defense use cases}
    \begin{threeparttable}
    \begin{tabularx}{0.45\textwidth}{Xccc}
        \toprule
        \multicolumn{4}{c}{\textbf{Watermark}}\\
        Dataset & \% Test Accuracy & \multicolumn{2}{c}{\% Trigger Accuracy}\\
        \cmidrule{3-4}
            &   & $\sigma_\text{invalid}$ & $\sigma_\text{valid}$\\
        \midrule
        CIFAR-10 & 92.26 & 14.00 & 100.00\\
        ImageNet & 80.11 & 0.00 & 100.00\\
        \midrule
        \end{tabularx}
        \begin{tabularx}{0.45\textwidth}{Xcc}
        \multicolumn{3}{c}{\textbf{User authentication}}\\
        Dataset & \multicolumn{2}{c}{\% Test Accuracy}\\
        \cmidrule {2-3}
            & $\textsf{sk}_\text{invalid}$ & $\textsf{sk}_\text{valid}$\\
        \midrule
        CIFAR-10 & 9.91 & 92.26\\
        ImageNet & 0.09 & 80.11\\
        \midrule
        \end{tabularx}
        \begin{tabularx}{0.45\textwidth}{Xcccc}
        \multicolumn{4}{c}{\textbf{Usage tracking}}\\
        Dataset & \multicolumn{3}{c}{\% Trigger Accuracy}\\
        \cmidrule{2-4}
            & $\overline{\text{Acc}}_{i,\mathcal{L}_i}$\tnote{1}& $\overline{\text{Acc}}_{i,\mathcal{L}_j}$\tnote{2} & $\max{\text{Acc}_{i,\mathcal{L}_j}}$\tnote{3}\\
        \midrule
        CIFAR-10 & \textbf{$100.00\pm 0.00$} & $10.04 \pm 0.06$ & $23.00$\\
        ImageNet & \textbf{$100.00\pm 0.00$} & $0.10\pm 0.01$ & $2.00$\\
        \bottomrule
        \end{tabularx}
        \begin{tablenotes}
            \item[1] $\overline{\text{Acc}}_{i,\mathcal{L}_i} = \frac{1}{n}\sum_{i=1}^{n}\text{Acc}(\mathcal{M}_i, X_t, \mathcal{L}_i)$: mean accuracy of model $\mathcal{M}_i$ (distributed to user $i$) on trigger set $X_t$ with label set $\mathcal{L}_i$.
            \item[2] $\overline{\text{Acc}}_{i,\mathcal{L}_j} = \frac{1}{n(n-1)}\sum_{i=1}^{n}\sum_{j=1,j\neq i}^{n}\text{Acc}(\mathcal{M}_i, X_t, \mathcal{L}_j)$: mean accuracy of model $\mathcal{M}_i$ on trigger set $X_t$ with label set $\mathcal{L}_j, j\neq i$.
            \item[3] $\max{\text{Acc}_{i,\mathcal{L}_j}} = \max_{i,j\neq i}\text{Acc}(\mathcal{M}_i, X_t, \mathcal{L}_j)$: maximum accuracy value of $\mathcal{M}_i$ on trigger set $X_t$ with label set $\mathcal{L}_j, j\neq i$.
        \end{tablenotes}
    \end{threeparttable}
    \label{table:defense}
\end{table}

\section{Conclusion}\label{sec:6}
In this work, we extend the black-box undetectable backdoor proposed from~\cite{goldwasserPlantingUndetectableBackdoors2022} to a specific use case, namely image classification task. We examine the effectiveness of this black-box signature-based backdoor for both attack and defenses. Regarding the attack, similar to the one proposed in~\cite{goldwasserPlantingUndetectableBackdoors2022}, we demonstrate that it is indeed powerful, undetectable and non-replicable for querying parties without a valid secret key. Nevertheless, the main downside of this attack lies in the significant computational cost of the signature/message decoding, which heavily depends on the digital signature scheme. As for the defenses, we propose three constructions - \textit{watermarking}, \textit{authentication} and \textit{usage tracking}, with provable undetectability, un-forgeabitlity and persistence under the assumption that the secret key is inaccessible to adversaries. Our experimental results show that the proposed schemes work effectively in assuring the privacy of the hosted NN:
\begin{itemize}
    \item \textit{Watermarking} enables a sole auditor with valid set of signatures to get the perfect trigger set accuracy, which is used to verify model ownership.
    \item \textit{Authentication} only allows legitimate users, with valid secret key, to query the real NN's outputs. Without the secret key, users only receive ``garbage'' outputs.
    \item \textit{IP tracking} ensures that each distributed model corresponds to just one secret key, making it easier to trace the source of privacy leak.
\end{itemize}

To our knowledge, our study is the first of its kind that empirically constructs a cryptographic backdoor parallelly to the host NN, which can be used for both beneficial and pernicious intentions. These implementations, still, demonstrate a few limitations, i.e. the considerable overhead of backdoor decoding in the backdoor attack and the strict requirement for black-box access for these applications to work. We believe that optimization tricks such as parallel computing or using more optimal steganography/digital watermarking techniques can help reduce the computation overhead significantly. Furthermore, as a future research direction, we aim to extend these schemes to white-box context. The idea of a white-box undetectable backdoor is also proposed in~\cite{goldwasserPlantingUndetectableBackdoors2022}, which, however is only applicable to models trained with Random Fourier Features learning paradigm~\cite{rahimiRandomFeaturesLargeScale2007}, which is not always relevant to modern ML context and is a good direction to explore.



\bibliography{IEEEabrv,refs}
%

\bibliographystyle{IEEEtran}











\newpage

 




\vfill

\end{document}